\documentclass[10pt,twocolumn]{IEEEtran}

\usepackage{cite}
\usepackage{amsmath,amssymb,amsfonts}
\usepackage{algorithmic}
\usepackage{graphicx}
\usepackage{textcomp}
\usepackage{xcolor}
\usepackage{booktabs}
\usepackage{multirow}
\usepackage{hyperref}
\usepackage{float}
\usepackage{caption}
\usepackage{amsthm}
\usepackage{epstopdf}

\newtheorem{theorem}{Theorem}[section]

\newtheorem{proposition}[theorem]{Proposition}
\newtheorem{corollary}[theorem]{Corollary}

\hypersetup{colorlinks=true,linkcolor=blue,citecolor=blue,urlcolor=blue}

\def\BibTeX{{\rm B\kern-.05em{\sc i\kern-.025em b}\kern-.08em
    T\kern-.1667em\lower.7ex\hbox{E}\kern-.125emX}}

\begin{document}

\title{Resolution Information: Limits of Ambiguity Resolution for Generative Communication}

\author{
\IEEEauthorblockN{
Angeles Vazquez-Castro\IEEEauthorrefmark{1},
Faheem Dustin Quazi\IEEEauthorrefmark{2},
Zhu Han\IEEEauthorrefmark{2}
}\\
\IEEEauthorblockA{\IEEEauthorrefmark{1}
Universitat Aut\`onoma de Barcelona (IEEC-UAB), Barcelona, Spain\\
Email: angeles.vazquez@uab.cat}\\
\IEEEauthorblockA{\IEEEauthorrefmark{2}
Electrical and Computer Engineering Department\\
University of Houston, Houston, TX, USA\\
Emails: fdquazi@CougarNet.UH.EDU, zhan2@central.uh.edu}
}

\maketitle

\begin{abstract}
In generative communication, the transmitter sends a compact generative
description, such as model parameters or a latent representation, rather than raw data. The receiver uses this description to form a
posterior belief over the underlying state and to resolve semantic ambiguity:
which interpretation, decision, or action is supported by the received
representation?
Inspired by Shannon's geometric view of communication as uncertainty
resolution, we introduce resolution information as the minimum information
update, measured in nats, required to move the receiver's posterior belief into
a low-ambiguity semantic region.  Our work yields three main results. First, when the receiver can form any posterior belief (the ideal unconstrained case), resolution information reduces to a binary divergence that depends only on each region's prior probability. In this case, the shape of the regions is irrelevant and under repeated sampling, ambiguity decays exponentially with an exponent equal to the resolution information, giving it an operational meaning as an \emph{ambiguity exponent}. Second, and most surprisingly, when the generative representation constrains
the posterior family, as in practice, geometry becomes operational and can
create irreducible ambiguity floors: half-spaces remain resolvable, whereas
polytope-type regions can exhibit residual ambiguity that no amount of
additional information can remove. These results reveal a fundamental departure from classical channel coding.
In Shannon theory, codes can be designed so that decoding regions separate
messages and error probability vanishes below capacity. In generative
communication, the model itself induces a constrained posterior geometry that
may prevent asymptotic ambiguity resolution. The resulting limit is not on rate,
but on resolvability itself.
\end{abstract}

\section{Introduction}

Classical information theory is built on the idea that information measures
the reduction of uncertainty. Entropy quantifies uncertainty, and mutual
information quantifies how much uncertainty is reduced by an observation
\cite{CoverThomas2006}. In Shannon's theory, this reduction of uncertainty
also has a geometric interpretation: messages, signals, and decoding regions
are represented as objects in suitable spaces, and reliable communication
corresponds to locating the received signal in the correct decoding
region~\cite{Shannon1949}.

Our observation is that generative communication can be cast on such theory by identifying the object of uncertainty. The receiver
is no longer necessarily trying to reconstruct a source sequence or decode a
discrete message. Instead, it receives a generative representation, for
example, a latent description, model parameter, learned prior, or compressed
semantic code, and uses it to form a posterior belief about an underlying
state of the world. The relevant question becomes whether the receiver's posterior belief is
sufficiently concentrated on the semantic region corresponding to the intended
interpretation, decision, or action.

This shift exposes a gap in the usual information-theoretic vocabulary.
Mutual information measures uncertainty reduction about random variables, and
rate--distortion theory measures the rate needed to reproduce a source within
a prescribed distortion. These concepts remain fundamental, but they do not
directly quantify the information needed to make a semantic interpretation
unambiguous under a posterior belief. In generative communication, the central
object is not necessarily a reconstructed signal, but a posterior distribution
over possible meanings.

Recent work on semantic and generative communication has developed several
approaches to quantifying and transmitting meaning. These include
task-oriented and semantic rate--distortion formulations
\cite{Liu2021,Qin2022,Guo2022}, posterior design and inference-aware
communication \cite{Akyol2026}, optimal-transport formulations for generative
communication \cite{Qu2024}, deep joint source--channel coding
\cite{Jin2025,Li2024,Nguyen2025}, and foundation-model-based generative
semantic communication \cite{Xu2024,Xu2025}. These approaches provide
important operational and algorithmic perspectives, typically optimizing
reconstruction quality, perceptual similarity, distribution matching, or
task-dependent empirical metrics.

In contrast, we take a different approach. We retain Shannon's fundamental principle: information measures the reduction of uncertainty, but apply it to a new object: ambiguity over \emph{semantic regions}. Formally, a semantic region is a measurable subset of the state space such as the collection of such regions forms a partition. Each region corresponds to states that share the same interpretation, but the geometry of these regions is entirely general and can be arbitrary. A generative representation is informative if it moves the receiver's belief toward a posterior that assigns high probability to one such region. Our central question is: how much information, measured in nats, is needed to reduce semantic ambiguity below a target level?

To answer this question, we introduce \emph{resolution information}. Informally,
if $\Gamma(p)$ denotes the ambiguity of a posterior belief $p$, and $p_0$ is
the receiver's initial belief, then the resolution information at ambiguity
level $\epsilon$ is
\begin{equation}
I_{\mathrm{res}}^*(\epsilon)
=
\inf_{p:\Gamma(p)\le \epsilon}
D_{\mathrm{KL}}(p\|p_0).
\label{eq:intro_resolution_information}
\end{equation}
Thus, $I_{\mathrm{res}}^*(\epsilon)$ is the minimum relative-entropy update
required to make the posterior sufficiently unambiguous. The full formal
definition of ambiguity and semantic regions is given in Section~\ref{sec:framework}.

Our formulation is grounded in the literature of information projections and divergence
geometry \cite{Csiszar1975,Csiszar2003,vanErven2014}, large deviations theory
\cite{DemboZeitouni1998}, and information geometry \cite{Amari2016}. It is
also inspired by the notion of channel resolvability introduced by Han and
Verd\'u \cite{HanVerdu1993} and later extended to quantum settings by Hayashi
\cite{Hayashi2006,Hayashi2011}. However, the object being resolved is different.
Classical channel resolvability concerns approximating output statistics using
coding or randomness resources. Here, generative resolvability concerns whether
a generative representation can drive posterior ambiguity over geometrical regions
to an arbitrarily small level.

Our results show that the distinction is significant. In classical channel coding, reliability is
limited by rate: below capacity, error probability can be driven to zero.
In generative communication, ambiguity may persist even when more information
is supplied, because the generative representation may restrict the posterior
family. In such cases, the limitation is not only on transmission rate but on
resolvability itself. The geometry of the semantic regions and the expressive
or precision limits of the posterior family can create irreducible ambiguity
floors. Our main contributions can be summarized as follows.

\begin{itemize}
    \item \textbf{Resolution information as ambiguity reduction.}
    We introduce resolution information as the minimum belief update needed to
    reduce posterior ambiguity below a target level. This gives a Shannon-like
    uncertainty-reduction principle for generative communication, where the
    object to be resolved is not a decoded message but a semantic region in the
    receiver's posterior belief space. The formal definition and basic
    properties are given in Section~\ref{sec:framework}.

    \item \textbf{Information-projection limit and ambiguity exponent.}
    In the ideal unconstrained geometrical space setting, we show that ambiguity resolution is an
    information-projection problem. The resulting limit depends only on the
    prior mass of the target semantic region, not on its geometric shape. We
    further show that resolution information governs the exponential decay of
    ambiguity under repeated sampling, giving it an operational meaning as an
    ambiguity exponent. These results are established in
    Section~\ref{sec:projection}.

    \item \textbf{Generative resolvability.}
    We define generative resolvability as the ability of a generative
    representation to drive posterior ambiguity to zero asymptotically. This
    notion parallels Shannon capacity, but captures a different limitation:
    ambiguity may persist not because the rate is too high, but because the
    representation cannot generate sufficiently resolving posteriors. This is
    developed in Section~\ref{sec:projection}.

    \item \textbf{Geometry-induced ambiguity floors.}
    We show that once the posterior family is constrained, geometry becomes
    operational. For Gaussian posterior families, half-spaces are resolvable,
    whereas polytope-type regions can exhibit irreducible ambiguity floors
    because they require simultaneous concentration along multiple directions.
    This geometric mechanism is analyzed in Section~\ref{sec:gaussian}.
\end{itemize}

Overall, our approach and results establish an information-theoretic framework for
ambiguity resolution in generative communication. They show that semantic
uncertainty can behave like classical uncertainty in the unrestricted regime,
where ambiguity decays exponentially, but can depart fundamentally from the
Shannon picture when posterior constraints create irreducible ambiguity.

The paper is organized as follows. Section~\ref{sec:framework} introduces the
generative communication framework, semantic ambiguity, and resolution
information. Section~\ref{sec:projection} derives the information-projection
limits, establishes the ambiguity-concentration exponent, and introduces
generative resolvability. Section~\ref{sec:gaussian} studies geometric limits
for Gaussian posterior families, showing that half-spaces are resolvable
whereas polytope-type regions can exhibit irreducible ambiguity floors.
Section~\ref{sec:conclusion} concludes the paper.


\section{Generative Communication Framework}
\label{sec:framework}

In this section we formalize generative communication as a belief-update problem as follows. The transmitter does not send raw observations. Instead, it sends a \emph{generative representation} $\theta$ (e.g. model parameters, latent description or a learned prior) that summarizes the information acquired from local observations. The receiver uses this representation to form a posterior belief over an underlying state space. The goal is not to reconstruct the original data, but to resolve \emph{ambiguity}: the receiver should concentrate enough posterior probability in the appropriate semantic region to support reliable interpretation.

Throughout thsi work, we assume that all probability measures are defined on a measurable space $(\mathcal{S}, \mathcal{B})$ and are absolutely continuous with respect to a common reference measure. We write $p$ for the corresponding density and $p(A) = \int_A p(s)\,ds$ for the probability assigned to a measurable set $A \subseteq \mathcal{S}$.

\subsection{Semantic Decision Regions and Ambiguity}

In analogy with decoding or decision regions in classical information theory,
we model semantic interpretations as measurable regions of the state space.
Each region contains the states that support the same interpretation, decision,
or action. Let $\mathcal{A}$ be a finite set of possible interpretations or actions. Each action $a \in \mathcal{A}$ is associated with a measurable \emph{semantic region}
\begin{equation}
A_a = \{ s \in \mathcal{S} : a \text{ is the appropriate interpretation for state } s \}.
\end{equation}
The collection $\{A_a\}_{a \in \mathcal{A}}$ is assumed to form a partition of $\mathcal{S}$ up to sets of measure zero. Thus, each state belongs to exactly one semantic region, except possibly on boundaries.

The receiver starts from an initial belief $p_0$ over $\mathcal{S}$. After receiving the generative representation $\theta$, the receiver forms a posterior belief
\begin{equation}
p(s) = p(s \mid \theta).
\end{equation}
For a semantic decision region $A_a$, the posterior probability assigned to that region is $p(A_a) = \int_{A_a} p(s)\,ds$.

The receiver must decide which interpretation is supported by the generative representation. From an information-theoretic perspective, the decision that minimizes residual uncertainty about the state is the one with the largest posterior probability \cite{CoverThomas2006}. We therefore define the \emph{ambiguity} of a posterior belief $p$ as
\begin{equation}
\Gamma(p) = 1 - \max_{a \in \mathcal{A}} p(A_a).
\label{eq:ambiguity}
\end{equation}
This quantity is the posterior probability mass outside the most likely semantic region. If almost all posterior mass lies in one region, then $\Gamma(p)$ is small and the interpretation is essentially resolved. If the posterior mass is spread across several regions, then multiple interpretations remain plausible and the ambiguity is large.

For some derivations, it is useful to focus on the simple case of one prescribed semantic region $A \subseteq \mathcal{S}$. In that case, the \emph{region-specific ambiguity} is
\begin{equation}
\Gamma_A(p) = 1 - p(A).
\end{equation}
Thus, requiring ambiguity at most $\epsilon$ is equivalent to requiring $p(A) \ge 1 - \epsilon$.

\subsection{Resolution Information}

We now quantify the information needed to move from the initial belief $p_0$ to a posterior belief $p$ with smaller ambiguity. Throughout the paper, logarithms are natural logarithms, so all information quantities are measured in nats.

Following classical information theory, the information gained by updating from $p_0$ to $p$ is measured by the Kullback--Leibler divergence \cite{CoverThomas2006}:
\begin{equation}
D_{\text{KL}}(p \| p_0) = \int_{\mathcal{S}} p(s) \log \frac{p(s)}{p_0(s)}\,ds.
\label{eq:kl}
\end{equation}
This quantity can be interpreted as the coding penalty incurred when using a code optimized for $p_0$ to describe a source distributed according to $p$, or equivalently, the information in nats needed to change one's belief from $p_0$ to $p$.

For a target ambiguity level $\epsilon \ge 0$, we define the \emph{resolution information} as the minimum information update required to achieve ambiguity at most $\epsilon$:
\begin{equation}
I_{\text{res}}^*(\epsilon) = \inf_{p : \Gamma(p) \le \epsilon} D_{\text{KL}}(p \| p_0).
\label{eq:resinfo}
\end{equation}
Intuitively, $I_{\text{res}}^*(\epsilon)$ is the smallest amount of information (in nats) that the generative representation must convey so that the receiver's posterior ambiguity is at most $\epsilon$. The constraint $\Gamma(p) \le \epsilon$ defines the set of beliefs that are sufficiently unambiguous. Resolution information is then the minimum relative-entropy distance from the initial belief to this set, or equivalently, the minimum additional information needed to resolve ambiguity to the desired level.

\subsection{Basic Properties}

Resolution information satisfies several immediate properties that formalize the intuition that easier ambiguity targets require less information, and that no information is needed if the initial belief already satisfies the target.

\begin{proposition}
\label{prop:basic}
The function $I_{\text{res}}^*(\epsilon)$ has the following properties:
\begin{enumerate}
\item \textbf{Monotonicity:} $I_{\text{res}}^*(\epsilon)$ is nonincreasing in $\epsilon$.
\item \textbf{Zero at feasibility:} If $\epsilon \ge \Gamma(p_0)$, then $I_{\text{res}}^*(\epsilon) = 0$.
\item \textbf{Positivity under separation:} If $\epsilon < \Gamma(p_0)$ and the feasible set $\mathcal{F}_\epsilon = \{p : \Gamma(p) \le \epsilon\}$ is separated from $p_0$ in total variation, i.e., there exists $\delta > 0$ such that 
  $d_{\text{TV}}(p, p_0) \ge \delta$ for all $p \in \mathcal{F}_\epsilon$, where 
  $d_{\text{TV}}(p, p_0) = \frac{1}{2} \int_{\mathcal{S}} |p(s) - p_0(s)|\,ds$, 
  then $I_{\text{res}}^*(\epsilon) > 0$.
\end{enumerate}
\end{proposition}

\begin{proof}
Monotonicity follows because the feasible set $\{p : \Gamma(p) \le \epsilon\}$ expands as $\epsilon$ increases. If $\epsilon \ge \Gamma(p_0)$, then the initial belief $p_0$ itself satisfies the ambiguity constraint, so $p = p_0$ is feasible and $D_{\text{KL}}(p_0 \| p_0) = 0$, yielding $I_{\text{res}}^*(\epsilon) = 0$. For the third property, if every feasible posterior satisfies $d_{\text{TV}}(p, p_0) \ge \delta$, then by Pinsker's inequality \cite{CoverThomas2006},
\begin{equation}
D_{\text{KL}}(p \| p_0) \ge 2 \, d_{\text{TV}}(p, p_0)^2 \ge 2\delta^2 > 0,
\end{equation}
and taking the infimum preserves the positivity.
\end{proof}

These properties are natural. Monotonicity says that allowing a larger ambiguity level cannot increase the required information. The second property says that if the initial belief is already sufficiently unambiguous, no additional information is needed. The third property says that if all low-ambiguity posteriors are separated from the initial belief in total variation, then a positive information update is necessary to achieve the ambiguity target.

Having established the basic framework, we now need to characterize the resolution information. The following section derives the fundamental limits, first for a fixed semantic region and then for a full partition.


\section{Fundamental Limits of Ambiguity Resolution}
\label{sec:projection}

We now derive the fundamental information-theoretic limits of ambiguity resolution. We first consider the case where the receiver can form \emph{any} posterior belief (the ideal, unconstrained setting). This yields a closed-form expression for resolution information. We then show that under repeated sampling, ambiguity decays exponentially with an exponent exactly equal to the resolution information, giving it an operational meaning as an ambiguity exponent. Finally, we introduce \emph{generative resolvability} as an analog of Shannon capacity and show that constrained posterior families can create an irreducible ambiguity floor, which implies a fundamental departure from classical channel coding.

\subsection{Projection onto a Fixed Semantic Region}

Let's consider a target semantic decision region $A \subseteq \mathcal{S}$.
Reducing the ambiguity associated with this region to at most $\epsilon$ is
equivalent to requiring $p(A)\ge 1-\epsilon$. We denote the corresponding
region-specific resolution information as
\begin{equation}
I_{\text{res}}^*(\epsilon; A) = \inf_{p : p(A) \ge 1 - \epsilon} D_{\text{KL}}(p \| p_0).
\label{eq:fixedregion}
\end{equation}
This is the minimum information needed to make the target region sufficiently probable under the posterior.

\begin{theorem}
\label{thm:fixedregion}
Let $A \subseteq \mathcal{S}$ be measurable with $0 < p_0(A) < 1$ and let $q = 1 - \epsilon$ be the target posterior probability on $A$. Then
\begin{equation}
I_{\text{res}}^*(\epsilon; A) = 
\begin{cases}
q \log \dfrac{q}{p_0(A)} + (1-q) \log \dfrac{1-q}{1-p_0(A)}, & q > p_0(A), \\[1em]
0, & q \le p_0(A).
\end{cases}
\label{eq:fixedregion_result}
\end{equation}
\end{theorem}

\begin{proof}
Every posterior $p$ induces a binary distribution $(p(A), 1-p(A))$ over $A$ and $A^c$. The initial belief $p_0$ induces $(p_0(A), 1-p_0(A))$. By the data-processing inequality for relative entropy under the mapping $s \mapsto \mathbf{1}\{s \in A\}$,
\begin{equation}
D_{\text{KL}}(p \| p_0) \ge D_{\text{KL}}\bigl((p(A),1-p(A)) \,\|\, (p_0(A),1-p_0(A))\bigr).
\label{eq:dpineq}
\end{equation}
If $q \le p_0(A)$, then $p_0$ already satisfies the constraint, so $I_{\text{res}}^*(\epsilon; A) = 0$.

If $q > p_0(A)$, the right-hand side of \eqref{eq:dpineq} is minimized by setting $p(A) = q$, because the binary divergence $d_{\text{bin}}(u \| r) = u \log(u/r) + (1-u)\log((1-u)/(1-r))$ is increasing in $u$ for $u > r$. This lower bound is achieved by the posterior
\begin{equation}
p^*(s) = 
\begin{cases}
\dfrac{q}{p_0(A)}\, p_0(s), & s \in A, \\[1em]
\dfrac{1-q}{1-p_0(A)}\, p_0(s), & s \in A^c.
\end{cases}
\label{eq:optimal_posterior}
\end{equation}
Since the optimal posterior only reweights the prior between $A$ and $A^c$; within
each region, the shape of $p_0$ is preserved. A direct substitution into the KL divergence yields
\begin{align}
D_{\text{KL}}(p^* \| p_0) &= \int_A p^*(s) \log \frac{p^*(s)}{p_0(s)}\,ds + \int_{A^c} p^*(s) \log \frac{p^*(s)}{p_0(s)}\,ds \\
&= p^*(A) \log \frac{q}{p_0(A)} + p^*(A^c) \log \frac{1-q}{1-p_0(A)} \\
&= q \log \frac{q}{p_0(A)} + (1-q) \log \frac{1-q}{1-p_0(A)},
\end{align}
which is the expression in \eqref{eq:fixedregion_result}.
\end{proof}

Note that Theorem \ref{thm:fixedregion} provides a baseline information-theoretic limit. This is because it shows that when posteriors are unrestricted, only the prior mass $p_0(A)$ matters and therefore the shape of $A$ is irrelevant. Moreover, the optimal update simply moves mass into $A$ with minimal relative-entropy change, rescaling the prior density by constants inside and outside the target region.

\begin{figure}[htbp]
\centering
\includegraphics[width=0.98\columnwidth]{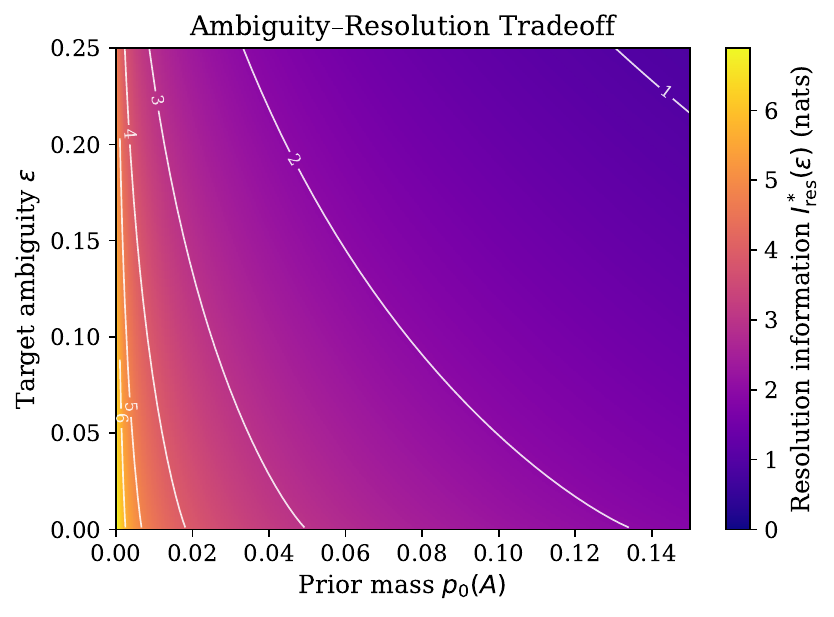}
\caption{Ambiguity--resolution tradeoff. The heatmap shows the resolution information
$I_{\mathrm{res}}^*(\epsilon;A)=d_{\mathrm{bin}}(1-\epsilon\|p_0(A))$
as a function of the prior mass $p_0(A)$ and the target ambiguity $\epsilon$.
White contours indicate level sets of constant $I_{\mathrm{res}}^*$.
The displayed window, $p_0(A)\in[10^{-3},0.15]$ and
$\epsilon\in[10^{-3},0.25]$, focuses on low-prior semantic regions and
stringent ambiguity targets, where the ambiguity--resolution tradeoff is
most pronounced. The resolution information is not bounded in general; it
diverges as $p_0(A)\to 0$ and $\epsilon\to 0$.}\label{fig:heatmap}
\end{figure}

Figure~\ref{fig:heatmap} visualizes the resulting ambiguity--resolution
tradeoff. It shows $I_{\mathrm{res}}^*(\epsilon;A)$ as a function of
the prior mass $p_0(A)$ and the target ambiguity $\epsilon$, with white
contours marking level sets of constant resolution information. The displayed
window focuses on low-prior semantic regions and stringent ambiguity targets,
where the tradeoff is sharpest: decreasing the tolerated ambiguity requires
increasingly larger resolution information, especially when $p_0(A)$ is small.

For instance, when $p_0(A)=0.1$, decreasing the target ambiguity from
$\epsilon=0.1$ to $\epsilon=0.01$ increases the resolution information from
approximately $1.76$ nats to $2.22$ nats. This illustrates that stricter
ambiguity requirements entail a higher resolution-information requirement.
In the unrestricted setting considered here, this tradeoff depends only on
the scalar prior mass $p_0(A)$; the geometry of the semantic region $A$ is
not yet visible. Section~\ref{sec:gaussian} shows that, once the posterior
family is constrained, the geometry of $A$ enters explicitly and may give
rise to irreducible ambiguity floors.

\subsection{Projection over a Decision Partition}

We now consider the full decision problem where the state space is partitioned into semantic regions $\{A_a\}_{a \in \mathcal{A}}$. The ambiguity constraint requires at least one region to have posterior probability at least $1 - \epsilon$:
\begin{equation}
\{p : \Gamma(p) \le \epsilon\} = \bigcup_{a \in \mathcal{A}} \{p : p(A_a) \ge 1 - \epsilon\}.
\label{eq:feasible_union}
\end{equation}
The feasible set is therefore a union of fixed-region feasible sets. The projection onto the union is obtained by choosing the closest of these sets.

\begin{corollary}
\label{cor:partition}
For a finite decision partition $\{A_a\}_{a \in \mathcal{A}}$,
\begin{equation}
I_{\text{res}}^*(\epsilon) = \min_{a \in \mathcal{A}} I_{\text{res}}^*(\epsilon; A_a),
\label{eq:partition_result}
\end{equation}
where each fixed-region term $I_{\text{res}}^*(\epsilon; A_a)$ is given by Theorem \ref{thm:fixedregion} with $p_0(A_a)$.
\end{corollary}

\begin{proof}
From \eqref{eq:feasible_union}, the ambiguity constraint defines a finite
union of region-specific feasible sets. Therefore, the projection
onto this union is obtained by choosing the closest such set:
\begin{align}
I_{\text{res}}^*(\epsilon) &= \inf_{p \in \bigcup_a \mathcal{F}_{a,q}} D_{\text{KL}}(p \| p_0) \\
&= \min_{a} \inf_{p \in \mathcal{F}_{a,q}} D_{\text{KL}}(p \| p_0) \\
&= \min_{a} I_{\text{res}}^*(\epsilon; A_a),
\end{align}
where $q = 1 - \epsilon$.
\end{proof}

This result gives the fundamental lower bound on the information needed to resolve ambiguity. The trade-off is clear: reducing $\epsilon$ (i.e., demanding higher posterior confidence) increases the required information. However, this is the unrestricted projection limit, achieved only when the receiver can choose any posterior belief. In practice, a generative representation restricts the set of achievable posteriors (e.g., to a parametric family), in different ways depending on the generative model. Consequently, the actual required information will be higher than this bound, and moreover some ambiguity targets may be unattainable.

\subsection{Exponential Ambiguity Concentration}

We now study how ambiguity decreases when the receiver obtains multiple independent samples from the generative representation. Unlike the one-shot projection limit, repeated sampling allows ambiguity to decay exponentially, analogous to error exponents in classical channel coding.

\begin{proposition}
\label{prop:sanov}
Suppose the receiver obtains $k$ independent samples from the generative representation, drawn i.i.d. according to the initial belief $p_0$, and forms the empirical posterior $p_k$. By Sanov's theorem \cite[Thm.~11.4.1]{CoverThomas2006}, the probability that the empirical distribution lies outside the set of beliefs satisfying the ambiguity constraint decays as
\begin{equation}
\Pr\bigl(\Gamma(p_k) > \epsilon\bigr) \le \exp\!\bigl(-k \cdot I_{\text{res}}^*(\epsilon) + o(k)\bigr).
\label{eq:sanov_decay}
\end{equation}
Consequently, the number of samples required to achieve $\Gamma(p_k) \le \epsilon$ satisfies
\begin{equation}
k_{\min} \ge \frac{1}{I_{\text{res}}^*(\epsilon)} \log\!\left(\frac{1}{\epsilon}\right) + o(1).
\label{eq:sample_complexity}
\end{equation}
\end{proposition}

\begin{proof}
Sanov's theorem states that for i.i.d. samples drawn from $p_0$, the probability that the empirical distribution falls into a set $\mathcal{F} \subseteq \mathcal{P}(\mathcal{S})$ is $\exp\!\bigl(-k \cdot \min_{p \in \mathcal{F}} D_{\text{KL}}(p \| p_0) + o(k)\bigr)$. Taking $\mathcal{F} = \{p : \Gamma(p) \le \epsilon\}$ as the set of low-ambiguity beliefs, the exponent is exactly $I_{\text{res}}^*(\epsilon)$. The event $\Gamma(p_k) > \epsilon$ is equivalent to the empirical distribution not belonging to $\mathcal{F}$, which gives the bound. Solving for $k$ yields \eqref{eq:sample_complexity}.
\end{proof}

Thus, the resolution information $I_{\text{res}}^*(\epsilon)$ directly determines the exponential decay rate of ambiguity, providing a clean operational meaning: it is the exponent of the fastest possible ambiguity decay when multiple samples are available.

\subsection{Generative Resolvability}

We now introduce \emph{generative resolvability} as the ability of a generative representation to drive ambiguity to zero as more semantic information is accumulated. We remark that this notion plays a role analogous to Shannon capacity, but in our context, ambiguity resolution replaces reliable decoding.

In this case, we simply use the exponential concentration law in Proposition \ref{prop:sanov}, to define generative resolvability as
\begin{equation}
C_{\text{gen}} = \sup \left\{ I_{\text{sample}} : \lim_{k \to \infty} \Gamma(p_k) = 0 \right\}.
\label{eq:genres}
\end{equation}
Thus, $C_{\text{gen}}$ captures the largest semantic information rate per sample under which ambiguity can be asymptotically resolved.

\begin{proposition}
\label{prop:genres}
Assume the exponential ambiguity concentration law \eqref{eq:sanov_decay}. Then:
\begin{enumerate}
\item \textbf{Ideal (unconstrained) case:} If the generative representation can achieve any positive $I_{\text{sample}}$, then $C_{\text{gen}}$ is unbounded; ambiguity can be driven arbitrarily close to zero.
\item \textbf{Semantic mismatch:} If the representation imposes a residual ambiguity floor $\alpha_{\min} > 0$, then
\begin{equation}
C_{\text{gen}} \le \frac{1}{c} \log\!\left(\frac{\Gamma(p_0)}{\alpha_{\min}}\right) < \infty,
\label{eq:genres_bound}
\end{equation}
and ambiguity cannot be resolved beyond this floor regardless of how much information is transmitted.
\end{enumerate}
\end{proposition}

\begin{proof}
In the ideal case, the exponential bound holds for any $I_{\text{sample}} > 0$, so $\Gamma(p_k) \to 0$ as $k \to \infty$.

With a residual floor $\alpha_{\min}$, the inequality $\Gamma(p_k) \le \Gamma(p_0) \exp(-c k I_{\text{sample}})$ cannot drive $\Gamma(p_k)$ below $\alpha_{\min}$. Solving for the limiting rate yields \eqref{eq:genres_bound}.
\end{proof}

Thus, semantic mismatch introduces a fundamental difference from Shannon capacity: a residual ambiguity floor $\alpha_{\min} > 0$ persists even with unlimited semantic information. This is different from classical channels, where error probability can be driven arbitrarily low below capacity; here, ambiguity cannot.

We have established the fundamental limits of ambiguity resolution when posteriors are unrestricted. The next section examines how these limits change when the generative representation constrains the posterior family, and demonstrates how geometry creates irreducible ambiguity floors.


\section{Geometric Limits: Gaussian Posterior Families}
\label{sec:gaussian}

The limit obtained in Theorem~\ref{thm:fixedregion} depends only on the prior masses $p_0(A_a)$ as the geometric shape of the semantic regions becomes invisible. This is because the receiver can choose \emph{any} posterior belief, including the optimal rescaling posterior in \eqref{eq:optimal_posterior}, which reshapes the prior arbitrarily.  However, in practice, a given generative representation restricts the set of achievable posteriors to a parametric family. This restriction clearly brings geometry into the projection and can create irreducible ambiguity floors.

We illustrate this phenomenon using Gaussian posterior families. Unlike in
classical communication theory, where Gaussian models often arise naturally
from noise and maximum-entropy arguments, Gaussianity is not intrinsic to generative
communication in general. However, here it is used as an analytically tractable
posterior family in which the key mechanism can be seen explicitly: once the
set of achievable posteriors is constrained, the geometry of the semantic
region can determine whether ambiguity is resolvable or whether an irreducible
floor appears. All logarithms in this section are natural, so divergences and
resolution information are measured in nats.

\subsection{Gaussian Belief Updates}

We assume here the initial belief is a multivariate Gaussian distribution expressed as follows
\begin{equation}
p_0 = \mathcal{N}(\boldsymbol{\mu}_0, \boldsymbol{\Sigma}_0),
\label{eq:prior_gaussian}
\end{equation}
with mean vector $\boldsymbol{\mu}_0 \in \mathbb{R}^d$ and positive definite covariance matrix $\boldsymbol{\Sigma}_0 \in \mathbb{R}^{d \times d}$. We restrict the posterior belief to the Gaussian family
\begin{equation}
p = \mathcal{N}(\boldsymbol{\mu}, \boldsymbol{\Sigma}),
\label{eq:posterior_gaussian}
\end{equation}
with $\boldsymbol{\mu} \in \mathbb{R}^d$ and $\boldsymbol{\Sigma} \succ 0$. The Kullback--Leibler divergence from $p$ to $p_0$ is
\begin{multline}
D_{\text{KL}}(p \| p_0) = \frac{1}{2} \Big[ \operatorname{tr}(\boldsymbol{\Sigma}_0^{-1}\boldsymbol{\Sigma}) + (\boldsymbol{\mu} - \boldsymbol{\mu}_0)^\top \boldsymbol{\Sigma}_0^{-1} (\boldsymbol{\mu} - \boldsymbol{\mu}_0) \\
- d + \log \frac{\det \boldsymbol{\Sigma}_0}{\det \boldsymbol{\Sigma}} \Big].
\label{eq:kl_gaussian}
\end{multline}

Note that this expression separates two distinct geometric mechanisms for reducing ambiguity: shifting the mean (changing the location of the belief) and shrinking the covariance (increasing its concentration).

\subsection{Half-Space Regions: Resolvable}

Consider a half-space semantic region
\begin{equation}
A = \{ \mathbf{s} \in \mathbb{R}^d : \mathbf{w}^\top \mathbf{s} \le T \},
\label{eq:halfspace}
\end{equation}
where $\mathbf{w} \in \mathbb{R}^d \setminus \{\mathbf{0}\}$ is the normal vector and $T \in \mathbb{R}$ is the threshold parameter. The boundary hyperplane is $\{ \mathbf{s} : \mathbf{w}^\top \mathbf{s} = T \}$. The pair $(\mathbf{w}, T)$ is defined only up to a positive scaling factor; the relevant geometric quantities are the unit normal $\mathbf{w} / \|\mathbf{w}\|$ and the distance $T / \|\mathbf{w}\|$ from the origin to the hyperplane. Under a Gaussian posterior $p = \mathcal{N}(\boldsymbol{\mu}, \boldsymbol{\Sigma})$, the projected variable $\mathbf{w}^\top \mathbf{s}$ is Gaussian with mean $\mathbf{w}^\top \boldsymbol{\mu}$ and variance $\mathbf{w}^\top \boldsymbol{\Sigma} \mathbf{w}$. Therefore,
\begin{equation}
p(A) = \Phi\!\left( \frac{T - \mathbf{w}^\top \boldsymbol{\mu}}{\sqrt{\mathbf{w}^\top \boldsymbol{\Sigma} \mathbf{w}}} \right),
\label{eq:halfspace_prob}
\end{equation}
where $\Phi$ is the standard normal cumulative distribution function. The argument of $\Phi$ is the Mahalanobis distance from the mean $\boldsymbol{\mu}$ to the boundary hyperplane, measured in units of the posterior standard deviation along $\mathbf{w}$.

A half-space is \emph{one-directional}: only the projected coordinate $\mathbf{w}^\top \mathbf{s}$ matters. To isolate the simplest geometric mechanism, first keep the covariance fixed at $\boldsymbol{\Sigma}_0$ and allow only the mean to shift: $p = \mathcal{N}(\boldsymbol{\mu}_0 + \boldsymbol{\Delta}, \boldsymbol{\Sigma}_0)$. Define the initial signed Mahalanobis distance to the boundary hyperplane as
\begin{equation}
\delta_0 = \frac{T - \mathbf{w}^\top \boldsymbol{\mu}_0}{\sqrt{\mathbf{w}^\top \boldsymbol{\Sigma}_0 \mathbf{w}}}.
\label{eq:delta0}
\end{equation}
If $\delta_0 \ge \Phi^{-1}(1-\epsilon)$, the initial belief already satisfies the constraint and no update is needed. Otherwise, the mean must shift toward the region along the normal direction. The minimum KL divergence is achieved by the displacement
\begin{equation}
\boldsymbol{\Delta}^* = -\bigl(\Phi^{-1}(1-\epsilon) - \delta_0\bigr) \frac{\boldsymbol{\Sigma}_0 \mathbf{w}}{\sqrt{\mathbf{w}^\top \boldsymbol{\Sigma}_0 \mathbf{w}}},
\label{eq:delta_star}
\end{equation}
and the resulting resolution information is
\begin{equation}
I_{\text{res}}^*(\epsilon; A) = \frac{1}{2} \left[ \bigl( \Phi^{-1}(1-\epsilon) - \delta_0 \bigr)^+ \right]^2,
\label{eq:halfspace_cost}
\end{equation}
where $(x)^+ = \max\{x,0\}$.

From this derivation we can observe two interesting aspects. First, by shifting the mean arbitrarily far into the half-space, the Mahalanobis distance $\delta_0$ can be made arbitrarily large and therefore $\epsilon$ can be driven arbitrarily close to zero. Second, the same effect can be achieved by reducing the variance along the direction $\mathbf{w}$ (i.e., increasing the Mahalanobis distance by shrinking the denominator) while keeping the mean fixed. Consequently, \emph{half-spaces are resolvable}: under Gaussian posterior families, ambiguity can be made arbitrarily small with sufficient information.

\subsection{Polytope Regions: Irreducible Ambiguity Floors}

Now consider a semantic region that requires simultaneous satisfaction of multiple directional constraints. A canonical example is the \emph{orthant-type polytope}
\begin{equation}
A = \{ \mathbf{s} \in \mathbb{R}^m : s_i \le a,\; i = 1,\ldots,m \},
\label{eq:polytope}
\end{equation}
where $a > 0$ is a common threshold and $s_i$ denotes the $i$-th coordinate of $\mathbf{s}$. This region is the intersection of $m$ half-spaces, one per coordinate. Unlike a single half-space, ambiguity resolution here requires the posterior to concentrate \emph{simultaneously} along all $m$ directions.

For simplicity, let's assume  an isotropic Gaussian posterior as follows
\begin{equation}
p = \mathcal{N}(\mathbf{0}, \sigma^2 \mathbf{I}_m),
\label{eq:isotropic_gaussian}
\end{equation}
and an initial belief $p_0 = \mathcal{N}(\mathbf{0}, \sigma_0^2 \mathbf{I}_m)$ with $\sigma_0 > \sigma$. Under this posterior, the coordinates are independent, and
\begin{equation}
p(A) = \Phi\!\left( \frac{a}{\sigma} \right)^m.
\label{eq:polytope_prob}
\end{equation}

We can notice that the quantity $a / \sigma$ admits a quite natural semantic interpretation: it is the \textbf{semantic margin}, i.e. the distance from the
posterior mean to each coordinate boundary, measured in posterior standard
deviations. A large margin means that the posterior is well concentrated inside
the admissible region; a small margin means that the posterior uncertainty
extends close to the boundary, leading to higher ambiguity.
We can then define the \textbf{maximum achievable semantic margin} as
\begin{equation}
\mu_{\max} \triangleq \frac{a}{\sigma_{\min}}.
\label{eq:mu_max}
\end{equation}
This quantity is the largest number of standard deviations that the constrained posterior can place within the acceptable region $A$ per coordinate, determined by the generative representation's precision limit $\sigma_{\min}$.

The target $p(A) \ge 1-\epsilon$ is therefore equivalent to
\begin{equation}
\Phi\!\left( \frac{a}{\sigma} \right)^m \ge 1-\epsilon,
\end{equation}
which means that
\begin{equation}
\frac{a}{\sigma} \ge \Phi^{-1}\bigl((1-\epsilon)^{1/m}\bigr),
\end{equation}
and therefore
\begin{equation}
\sigma \le \sigma_{\star} \triangleq \frac{a}{\Phi^{-1}\bigl((1-\epsilon)^{1/m}\bigr)}.
\label{eq:polytope_constraint}
\end{equation}

The resolution information is achieved by reducing the variance from $\sigma_0^2$ to the smallest admissible value that satisfies the constraint. If the generative representation imposes a \emph{minimum variance} $\sigma_{\min}^2 > 0$ (e.g., due to limited precision, finite model capacity, or architectural constraints), then the maximum achievable semantic margin is $\mu_{\max} = a / \sigma_{\min}$, and the maximum achievable mass on $A$ is
\begin{equation}
p_{\max}(A) = \Phi\!\left( \mu_{\max} \right)^m.
\label{eq:pmax}
\end{equation}
The corresponding minimum achievable ambiguity is
\begin{equation}
\epsilon_{\min} = 1 - \Phi\!\left( \mu_{\max} \right)^m.
\label{eq:epsilon_floor}
\end{equation}

Hence, if  $\epsilon < \epsilon_{\min}$, the target ambiguity lies below what the
constrained Gaussian family can attain. In this case, no admissible posterior
satisfies the ambiguity constraint, and the corresponding resolution information
is infinite. The quantity $\epsilon_{\min}$ is therefore an irreducible
ambiguity floor: it is the \textbf{residual ambiguity imposed by the finite precision of
the generative representation.}

Figure~\ref{fig:polytope_floor} visualizes this floor as a function
of the dimension $m$ and the maximum semantic margin $\mu_{\max}$. The plot
shows the two competing effects in \eqref{eq:epsilon_floor}: increasing $m$
raises the ambiguity floor because more constraints must be satisfied
simultaneously, whereas increasing $\mu_{\max}$ lowers the floor by allowing
the posterior to concentrate more deeply inside the admissible region. Thus,
bounded precision is especially restrictive for high-dimensional semantic
regions.

\begin{figure}[htbp]
\centering
\includegraphics[width=0.95\columnwidth]{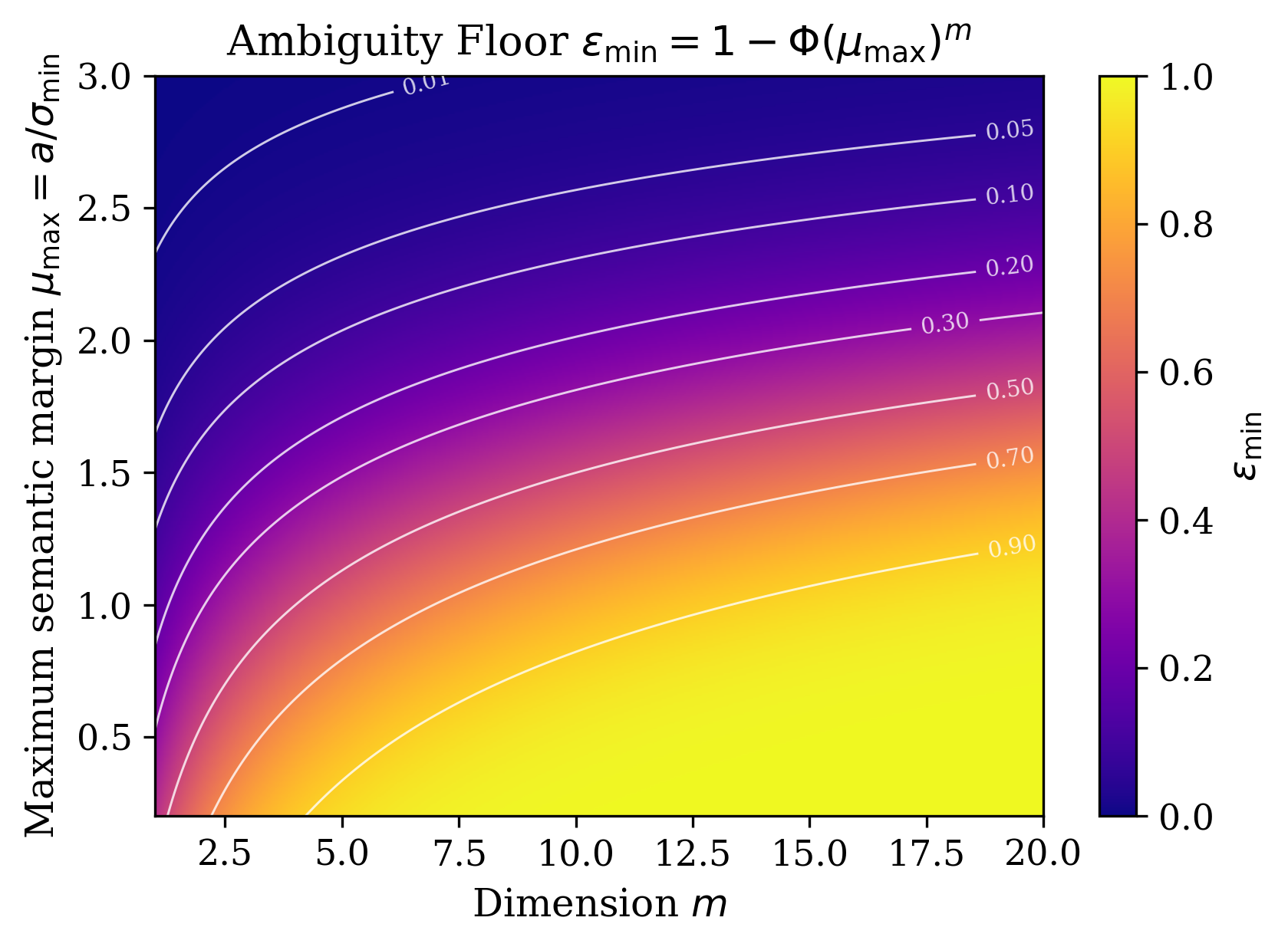}
\caption{Ambiguity floor for the polytope $A=\{s_i\le a\}$ under Gaussian posteriors. The heatmap shows $\epsilon_{\min}=1-\Phi(\mu_{\max})^m$ as a function of the dimension $m$ and maximum semantic margin $\mu_{\max}$. Larger $m$ increases the floor, while larger $\mu_{\max}$ reduces it.}
\label{fig:polytope_floor}
\end{figure}

\begin{figure}[htbp]
\centering
\includegraphics[width=0.95\columnwidth]{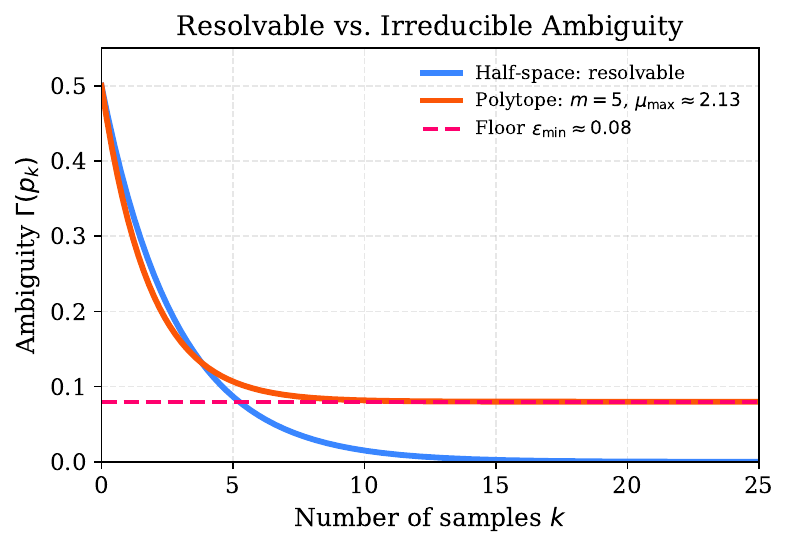}
\caption{Resolvable and non-resolvable geometries under constrained Gaussian posteriors.
The half-space case (blue) is resolvable, with ambiguity decaying to zero.
The polytope case (orange), with $m=5$ and $\mu_{\max}\approx 2.13$,
exhibits an irreducible ambiguity floor
$\epsilon_{\min}=1-\Phi(\mu_{\max})^m\approx 0.08$ (pink dashed).}
\label{fig:halfspace_vs_polytope}
\end{figure}

\subsection{Generative Resolvability for Gaussian Families}
\label{sec:gaussian_resolvability}

The ambiguity floor in \eqref{eq:epsilon_floor} translates directly into a
limit on generative resolvability. For the constrained Gaussian family with
minimum variance $\sigma_{\min}^2>0$, the residual ambiguity is
\begin{equation}
\alpha_{\min}
=
\epsilon_{\min}
=
1-\Phi(\mu_{\max})^m,
\label{eq:gaussian_alpha_min}
\end{equation}
where $\mu_{\max}=a/\sigma_{\min}$ is the maximum achievable semantic margin.
Substituting this floor into Proposition~\ref{prop:genres} gives
\begin{equation}
C_{\mathrm{gen}}
\le
\frac{1}{c}
\log\!\left(
\frac{\Gamma(p_0)}
{1-\Phi(\mu_{\max})^m}
\right)
<\infty,
\label{eq:genres_gaussian}
\end{equation}
whenever the floor is positive.

Thus, finite posterior precision implies finite generative resolvability.
For any finite $\mu_{\max}$, the posterior family cannot concentrate enough
mass inside the polytope to eliminate ambiguity completely. Increasing
$\mu_{\max}$ lowers the floor and increases $C_{\mathrm{gen}}$, while
increasing the dimension $m$ has the opposite effect because more constraints
must be satisfied simultaneously. In the ideal limit
$\sigma_{\min}\to 0$, equivalently $\mu_{\max}\to\infty$, the floor vanishes
and the unconstrained behavior is recovered.

Figure~\ref{fig:halfspace_vs_polytope} illustrates the resulting contrast between a
resolvable half-space and a polytope with an irreducible floor.

\subsection{Effective Dimensions}
\label{sec:geometric_interpretation}

Our results from the Gaussian example suggests a simple geometric principle. Specifically, we obtained that in the unrestricted
setting of Theorem~\ref{thm:fixedregion}, geometry is invisible because the
posterior can be reshaped arbitrarily. On the contrary, once the posterior family is constrained,
the geometry of the semantic region matters through the number of independent
directions in which the posterior must concentrate.

Moreover, a half-space is one-directional: only the projection onto its normal vector
matters. By contrast, the orthant polytope
$A=\{\mathbf{s}:s_i\le a,\; i=1,\ldots,m\}$ is multi-directional: it requires
simultaneous concentration along $m$ coordinate directions. Under the isotropic
Gaussian model considered above, this leads to the floor
\begin{equation}
\epsilon_{\min}
=
1-\Phi(\mu_{\max})^{m}.
\end{equation}

More generally, one may interpret $m$ as an effective dimension
$d_{\mathrm{eff}}(A)$: the number of independent directions along which the
posterior must be sufficiently concentrated.

Note that this interpretation should be understood only relative to the posterior family assumption. In our assumed
 isotropic case, each effective direction has the same maximum semantic
margin $\mu_{\max}$, which yields the simple expression
$1-\Phi(\mu_{\max})^{d_{\mathrm{eff}}(A)}$. For anisotropic or correlated
posterior families, the ambiguity floor depends on the directional margins and
correlations, rather than on $d_{\mathrm{eff}}(A)$ alone. In any case, the insight is
that generative resolvability depends on the alignment between the geometry of
the semantic task and the concentration directions available to the generative
model.

\section{Conclusion}
\label{sec:conclusion}
We introduced resolution information as the minimum belief update needed to
reduce semantic ambiguity in generative communication. In the ideal
unconstrained case, the receiver can form any posterior belief. In this regime,
resolution information reduces to a binary divergence: only the prior mass of
the target semantic region matters, and the geometry of the region is invisible.
As a result, ambiguity can be driven arbitrarily close to zero.

The situation changes when the posterior is constrained by a realistic
generative representation. In that case, the geometry of the semantic region
matters. Some regions can still be resolved, but others create irreducible
ambiguity floors. This is a fundamental difference from classical channel
coding. In Shannon theory, error probability can vanish below capacity because
codes and decoding regions can be designed freely enough. In generative
communication, the model itself restricts the posterior beliefs that the
receiver can form. Therefore, ambiguity may persist even with additional
information. The limitation is not only a limit on rate, but a limit on
resolvability.

Our Gaussian analysis makes this mechanism explicit. For half-spaces, ambiguity
can be resolved because only one direction must be controlled. For
polytope-type regions, several directions must be controlled at the same time.
If the generative representation has a finite precision limit, the posterior
cannot concentrate arbitrarily well along all these directions. The resulting
ambiguity floor is governed by the effective dimension of the semantic region
and by the maximum semantic margin
$\mu_{\max}=a/\sigma_{\min}$. Thus, ideal asymptotic resolvability fails when
the model cannot produce sufficiently precise posteriors in all directions
required by the semantic task.

These results suggest several directions for future work. One direction is to
characterize ambiguity floors for learned generative models, such as VAEs,
diffusion models, and LLMs. Another is to quantify the geometric complexity of
semantic regions beyond the simple half-space and polytope examples studied
here. A third direction is to design generative representations that reduce
posterior uncertainty along the directions that matter most for semantic
resolution. More broadly, the framework suggests that improving generative
communication may require not only better compression or reconstruction, but
also better alignment between the geometry of the semantic task and the
posterior beliefs that the model can represent.



\end{document}